\documentclass[onecolumn]{svjour3}          
\smartqed  
\usepackage{graphicx}

\usepackage{bm,bbm,amsmath,amsfonts,amssymb}
\usepackage{arcs}
\usepackage{algorithm}
\usepackage{algorithmic}
\usepackage{tcolorbox}
\usepackage{mathtools}
\usepackage{bmpsize}

\DeclarePairedDelimiter\ceil{\lceil}{\rceil}
\DeclarePairedDelimiter\floor{\lfloor}{\rfloor}

\begin{document}

\title{Consensus with Preserved Privacy against Neighbor Collusion\thanks{This work was supported by Knut and Alice Wallenberg Foundation.}}




\author{Silun Zhang         \and 
		Thomas Ohlson Timoudas					\and
        Munther A. Dahleh      
}


\institute{S. Zhang \at
           Laboratory for Information and Decision Systems (LIDS), MIT, MA 02139, USA
 \\
              \email{silunz@mit.edu}           
           \and
           T. Ohlson Timoudas\at
           KTH Royal Institute of Technology
           100 44 Stockholm, Sweden\\
           \email{ttohlson@kth.se} 	
           \and
           M. Dahleh \at
              Laboratory for Information and Decision Systems (LIDS), MIT, MA 02139, USA\\
              \email{dahleh@mit.edu}  
}

\date{Received: date / Accepted: date}

\maketitle

\begin{abstract}
This paper proposes a privacy-preserving algorithm to solve the average consensus problem based on Shamir's secret sharing scheme, in which a network of agents reach an agreement on their states without exposing their individual state until an agreement is reached. Unlike other methods, the proposed algorithm renders the network resistant to the collusion of any given number of neighbors (even with all neighbors' colluding). Another virtue of this work is that such a method can protect the network consensus procedure from eavesdropping. 



\end{abstract}

\section{Introduction}
\label{intro}

As a successful conceptual abstraction of many emerging phenomena  in nature, the multi-agent model (MAM) has been extensively studied in the last decades. 
These studies initiated from the fundamental but also inspiring problem, namely multi-agent consensus, which aims at driving the states of agents to a global agreement. The formal study of the consensus problem was first introduced by Degroot \cite{degroot1974reaching} in the 1970s, which sparked many further extensions and developments (see, e.g., \cite{olfati2007consensus,ren2005survey,jadbabaie2003coordination} and the references therein).
The methods and results, that were conceived in the study of the consensus problem, were later adapted to network coordination problems in control theory, such as formation control \cite{song2017intrinsic,zhang2020intrinsic,fax2004information,zhang2018intrinsic,zhang2016spherical}, distributed optimization \cite{nedic2014distributed,nedic2009distributed,nedic2010constrained,qu2017harnessing},  and also network games \cite{cenedese2020asynchronous,zhang2019strategy}. 
In the control setting, the situation is further complicated, as the the agents are governed by  nonlinear dynamics \cite{thunberg2014distributed,zhang2020modeling}.

In distributed algorithms, each agent typically requires access to the states of its neighbors in order to compute a local update. 
This may leave the agents in the network vulnerable, as some of them may not wish to disclose this local information to their neighbors, especially if some of it is highly private or sensitive.
For many applications, it is therefore essential to achieve network consensus, while preserving the privacy of each agent.
For example, in opinion dynamics \cite{albi2014boltzmann}, opinions may relate to a sensitive topic, and the participating individuals expect these to remain secret, especially from their acquaintances, i.e., each of their neighbors.




The fundamental idea of many existing approaches to privacy-preserving consensus algorithms, is to conceal the individual state by adding a deterministic or stochastic disturbance to the real state before communicating it.
This idea was first mentioned by Kefayati et al. in \cite{kefayati2007secure}, where a zero-mean normal noise was added to the agents' state. Based on this work, \cite{huang2012differentially} proposed a synchronization algorithm that blends the true state with a random noise drawn from a Laplace distribution with a time-decaying magnitude. 
These methods guarantee that the agents reach consensus, but the value agreed upon may not necessarily be the average of the initial states.
To achieve average consensus, Mo and Murray \cite{mo2016privacy} shifted the real state by a particular linear combination of Gaussian processes, tailored specifically to ensure that at any time the sum of all the noise injected previously  vanishes exactly (also cf. \cite{liu2017secure,he2018privacy,katewa2015protecting}).
The true states can also be masked by using a deterministic state mapping \cite{altafini2019dynamical}, or by simply adding a deterministic disturbance, such as an offset \cite{manitara2013privacy,gupta2017privacy}, and a perturbation function whose integral corresponding to the historical effects equals to zero \cite{rezazadeh2018privacy}.

Another approach to privacy-preserving consensus takes advantage of  \textit{homomorphic encryption} schemes, which allows  algebraic computations to be performed directly on the encrypted data without the need of deciphering it first \cite{gentry2009fully}. Indeed, many prevalent encryption methods are naturally   \textit{partially homomorphic}, meaning that one, but not the other, of the addition and multiplication operations
can be performed directly on the ciphertexts. For example, the \textit{RSA} and \textit{ElGamal} cryptosystems allow multiplications on the ciphered data, and the \textit{Benaloh} and \textit{Paillier} systems allow additions  without deciphering \cite{acar2018survey}. In the classical  consensus algorithm, only addition is involved, which is why the Paillier cryptosystems were used in \cite{alexandru2019encrypted} and \cite{ruan2019secure}  to achieve privacy-preserving consensus.

In cryptography, the \textit{secret sharing scheme} is a multi-party encryption method to share a confidential message with multiple parties, that ensures that even with the collusion of a certain number of parties, it is still not possible to uncover the secret message.  
Shamir \cite{shamir1979share} presented the first algorithm to solve the secret sharing problem, and the good survey paper  \cite{beimel2011secret} provides more insights on the historical evolution of this problem.


In this work, we employ a secret sharing scheme for the communication between agents, to address the problem of privacy-preserving consensus on undirected graphs.
 As in network coordination, each agent is required to send its state to all its neighbors, which indeed injects into the network the duplicates of same piece of information. 
In order to preserve its privacy, each agent can partition its local information into several so-called secret shares. 
 Then, rather than sending the full information to all of its neighbors, each agent only sends one share to each neighbor. As the information can be reconstructed entirely from a certain specified number of shares, the information injected into the network for each agent is still intact but in a confidential way. 
 In addition, the secret sharing scheme makes the communication naturally resistant to eavesdropping.

Unlike methods based on \textit{differential privacy} techniques  \cite{huang2012differentially,mo2016privacy,liu2017secure,he2018privacy,katewa2015protecting}, the proposed method can reach average consensus with no errors, and also  protect the network from eavesdropping. 
In addition, the privacy security adopted in this paper  renders the network immune to the collusion of any given number of neighbors. This is in contrast to \cite{manitara2013privacy,mo2016privacy}, where at least one   neighbor of each agent must be honest. 


In the rest of the paper, we use  $\mathbb Z+$ to denote the set of all positive integers, and for any $N \in \mathbb Z+$, we define the set $[N]=\{1, 2, \dots, N\}$. In addition, we denote by $|S|$ the cardinality of a given set $S$. For any event $A$, the indicator function $\mathbbm{1}(A) =1 $ when $A$ happens, otherwise $\mathbbm{1}(A) =0 $. Moreover, we denote $\mathbf 1 \in \mathbb R^n$ is the vector consisting of all one entries.

\section{Preliminaries}
\label{sec:prelimi}
In this section, we will introduce the notions used in the paper, and revisit some fundamental results on graph theory and secret sharing schemes.

\subsection{Network graph and consensus algorithm with switching topology}
For a networked system, the topology of inter-agent connectivity can be modeled by a graph $\mathcal{G}=(\mathcal{V},\mathcal{E})$, where the set of nodes is $\mathcal{V}=\big\{1,\ldots,N\big\}$, and $\mathcal{E}\subset \mathcal{V} \times \mathcal{V}$ is the edge set. 
A graph $\mathcal G$ is undirected if $(i,j)\in \mathcal E$, for any $(j,i) \in \mathcal E$. 
In this paper, without further indication, we assume that all the graphs are undirected. We also define the neighbor set of a node $i$ as $\mathcal{N}_i= \big\{ j:(j,i) \in \mathcal{E}\big\}$, and we say that $j$ is a neighbor of $i$, if $j \in \mathcal{N}_i$. 
Moreover, we say that two edges are adjacent if they are incident to a same endpoint.

The union of any two graphs $\mathcal G_1$ and $\mathcal G_2$ is defined by $\mathcal G_1 \bigcup \mathcal G_2=\big(\mathcal V_1 \bigcup\mathcal V_2 ,  \mathcal E_1\bigcup \mathcal E_2\big)$. 
Given a finite set of graphs $\bar{\mathcal{G}} = \{\mathcal G_1, \mathcal G_2, \dots, \mathcal G_M\}$ all having the same set of nodes $\mathcal V$, we say that a function  $\mathcal{H}: \mathbb Z+ \to \bar{\mathcal{G}}$ is a dynamical graph. 
Moreover, a dynamical graph $\mathcal{H}(t)$ is \textit{jointly connected} across any time interval $I \subset \mathbb Z+$ if  the graph $\bigcup_{t\in I}  \mathcal{H}(t)$ is connected.

Without loss of generality, we denote by $x_i(t) \in \mathbb R$ the state of node $i$ at time $t\in \mathbb Z+$. We say that the agents in such a network have reached consensus, if for any initial condition $x_i(0)$, $i \in \mathcal V$,  it holds that at some time $t \in \mathbb Z+$, $x_i(t)=x_j(t)$ for all $i,j \in \mathcal V$. 
To tackle this problem, one can cast to the discrete-time consensus algorithm 
\begin{equation}\label{eq:consensus_basic}
x_i(t+1)=x_i(t) + \alpha_i(t) \sum_{j\in \mathcal N_i} (x_j(t)-x_i(t)),
\end{equation}
for $i\in [N]$, where $\alpha_i(t) \in \mathbb R$ is a step size. It has been shown that if a communication graph $\mathcal G$ is connected, and the step size $\alpha_i$ in \eqref{eq:consensus_basic} satisfies $\alpha_i < \frac{1}{ |\mathcal N_i|}$, then the iterative algorithm given in  \eqref{eq:consensus_basic} ensures that  the states of all agents converge to the average value $\bar x_0$, where $\bar x_0=\frac{1}{N} \sum_{i=1}^N x_i(0)$, (see \cite{olfati2007consensus,jadbabaie2003coordination,ren2005consensus} and the references therein).
 
Moreover, when the  communication topology is a dynamical graph $\mathcal H(t)$, $t\in \mathbb Z+$, the following lemma gives a sufficient condition for asymptotic convergence of the consensus algorithm in \eqref{eq:consensus_basic}.

\begin{lemma}[Proposition 2, \cite{moreau2005stability}]\label{lem:time-switching_nonaverage_consensus}
Given a dynamical undirected graph $\mathcal H(t)$, let $\alpha_i(t)=\frac{1}{1+|\mathcal N_i(t)|}$ for each time $t$. If for any time $T$ the graph $\bar{\mathcal{H}}_T = \bigcup_{t=T}^\infty \mathcal H(t)$ is connected, then consensus is globally asymptotically reached using the iterative algorithm in \eqref{eq:consensus_basic}.
\end{lemma}

Note that Lemma~\ref{lem:time-switching_nonaverage_consensus} only guarantees  consensus but not average consensus for the agents' states.

\subsection{Secret sharing schemes}
The secret sharing scheme is 
an encryption method for sharing a confidential message with multiple parties, such that even with the collusion of a certain number of parties, the message should still not be  disclosed. 

Specifically, an $(n,p)$ secret sharing scheme consists of two algorithms (\textbf{Share}, \textbf{Reconstruct}) with the forms that
\begin{itemize}
\item \textbf{Share} takes as input a secret $M$ and outputs $n$ shares $(M_1,\dots, M_n)$;
\item \textbf{Reconstruct} takes as input $p$ different shares $(M_i)_{i\in \mathcal I}$ for any index set $\mathcal I \subset [n]$ with $|\mathcal I|=p$, and outputs $M$.
\end{itemize}

The generated secret share $M_i$ is then  distributed to the party $i$ for each $i\in[n]$. The security of such a scheme requires that any collusion of less than $p$ parties should reveal no information about the message $M$. More precisely, this means that for any index set $\mathcal I\subset [n]$ with $|\mathcal I|<p$, the distribution of $(M_i)_{i\in \mathcal I}$ should be  independent of the true message $M$. Secret sharing schemes are used in many applications, e.g., encryption keys, distributed storage, missile launch codes, and numbered bank accounts. In these applications, each of the generated pieces of information must keep the original message confidential, as their exposure is undesirable, however, it is also critical that the message should not be lost. 

One celebrated secret  sharing scheme is the Shamir's scheme proposed by Adi Shamir in 1979 \cite{shamir1979share}. In this scheme, the Share algorithm samples the values of a secret $(p-1)$-order polynomial at $n$ different points, and  
the Reconstruct algorithm can recover the secret polynomial from any $p$ of these samples. 
In addition, as the $(p-1)$-order polynomial contains the true message $M$ as its constant term, $M$ can be   reconstructed (see Appendix for more details). 

%


\section {Problem formulation}
In this paper, we want to achieve network consensus subject to  communication safety, and also anti-collusion of neighbors.

Define the security degree by an integer tuple 
\[p=(p_1, p_2, \dots, p_N)\in \mathbb Z+^{N},\]
with $(p_i-1)$ indicating the maximum number of  neighbors of agent $i$ that are allowed to collude without a privacy leak.
Next, we give the detailed definition of the security adopted in the paper.

\begin{definition}[$p$-degree security]\label{def:p-security}
We say that an algorithm in a network is of \textit{$p$-degree security} if for each agent $i\in[N]$,  at any time $t\in \mathbb Z+$,
\begin{enumerate}
\item[(1)]  the state $x_i(t)$ is safe even with the collusion of less than $p_i$ neighbors in $\mathcal N_i$, 
and
 
\item[(2)] it is not possible to disclose the state $x_i(t)$ by eavesdropping the communication on less than $p_i$ edges in $\{(i,j)\in \mathcal E: j\in \mathcal N_i\}$.
\end{enumerate}
\end{definition}

In this definition, we set the security degree $p_i$ for each agent $i$ is to assure that even a certain number of $i$'s neighbors betray or get attacked, the state $x_i(k)$ is still not leaked.

\begin{remark}
For an algorithm of $p$-degree security,
(i) if $p_i\geq 2$, the state $x_i(t)$ is kept secret from  every neighbor $j\in \mathcal{N}_i$. 
(ii) If $p_i=|\mathcal N_i|$, then $x_i(t)$ is disclosed only when all neighbors of a node are colluding. (iii) If $p_i > |\mathcal N_i|$, then the state $x_i(t)$ is completely confidential in the network. 
\end{remark}



Now we are ready to state the problem of privacy-preserving consensus that will be solved in this paper.

\begin{problem}\label{prob:1}
In a network consisting of $N$ agents, the problem of \textit{privacy-preserving consensus} with  security degree $p$ is to achieve   for every agent $i$ that
\begin{enumerate}
\item[(a)] average consensus is reached, i.e.,
\[
\lim_{t\to \infty} x_i(t)=\frac{1}{N}\sum_{j\in[N]} x_j(0), \qquad \forall i \in [N], \text{ and }
\]

\item[(b)] the consensus algorithm is of   $p$-degree security. 
\end{enumerate} 
\end{problem}

We note that the privacy-preserving requirement of $(2)$ in Problem~\ref{prob:1} is only valid before  consensus has been reached. After that, although the information transmitted between agents is still encrypted, or more precisely the communication satisfies $p$-degree security, the agent state $x_i(k)$ is nevertheless already known to all the other agents, due to the state consensus.

\section{Privacy-preserving consensus based on secret sharing}

%
%
%

Two algorithms are proposed in this section. The first one solves the privacy-preserving consensus problem by using a secret sharing scheme in  communication. 
The second one is a 
key distribution algorithm which synchronizes a secret key across the network within finite steps.

\subsection{Privacy-preserving consensus algorithm}\label{sec:PPC}
In this subsection, we propose an algorithm inspired by the Shamir's secret sharing scheme to solve Problem~1.

As a private key, each agent $i$  will randomly initialize\footnote{In this paper, any random  variables are drawn from the uniform distribution on their supported sets.}
 a coefficient vector $a^{(i)}=(a_1^{(i)}, \dots, a^{(i)}_{p_i-1}) \in \mathbb [-1, 1]^{p_i-1}$, and  keep this vector 
$a^{(i)}$ secret from all other agents. Then for each $i\in[N]$, we define an encryption  polynomial of order $(p_i-1)$ by
\[
	f_i(\theta,t)=\sum_{j=1}^{p_i-1} a_j^{(i)} \theta^j + x_i(t),
\]
for $\theta \in \mathbb R$, $t\in \mathbb Z+$. This polynomial is  only known to agent $i$, since  the coefficients $a^{(i)}$ and state $x_i(t)$ are both hidden from the others. Note that the state $x_i(t)$ is equal to $f_i(0,t)$.

Let a security degree $p \in \mathbb Z+^N $, and we say that $p$ has  maximal order  $\bar p$,  if it holds that
\[\bar{p}=\max_{i \in [N]} p_i.\]
Note that such an upper bound $\bar p$ of  secret degrees  indicates the maximal capacity of the tolerance for    attacks on the neighbors, and can be preassigned for a privacy-preserving algorithm.

Then we define the \textit{key sequence} for $\bar p$ communication channels by an integer vector 
\[S=(S_1, \dots, S_{\bar p}) \in [1, \kappa]^{\bar p}\subset \mathbb Z^{\bar p},
\]
where $\kappa  \in \mathbb Z+$ is the maximal possible key and in general  $\kappa \gg \bar p$.


In order to ensure that the agents communicate using the same key sequence, they first need to agree on a common key sequence $S$.
To this end, one method is to preassign a random, or a default key sequence for all the agents. For example, we can set a default key $S=(1, 2, 3, \dots, \bar p)$, which is actually  widely used in many secret sharing applications. 
An alternative method to establishing  a common  key sequence $S$   is to use a consensus algorithm, which  can synchronize the key sequences of all the agents starting from any random initial keys.
The latter will be further detailed in Section \ref{sec:key_distribution}.

Next, we present Algorithm~\ref{alg:moment_method} to reach   asymptotic consensus with   $p$-degree security for a group of agents $\mathcal V$. 
To store the information transmitted through each channel, every node $i\in \mathcal V$ sets a local buffer $r_i(k) \in \mathbb R$, for every $k\in S$. 

Algorithm 1 starts with  a Handshake procedure (Step 1 in Algorithm~\ref{alg:moment_method}), in which a communication channel $c_{ij} \in S$ is assigned for each edge $(i,j)$. Particularly, at the beginning of each iteration, the agent $i$, $i\in \mathcal V$, randomly selects a Handshake Pair $(b_j^{(i)},c^{(i)}_j)$ for each  edge $(i,j)\in \mathcal E$, with $b_j^{(i)}\in [-1,1]$ being a random bid value, and $c^{(i)}_j \in S$ 
being the preferred channel for edge $(i,j)$ according to agent $i$.
Then the edge $(i,j)$ will be assigned  
the preferred channel with the larger bid value from the nodes $i$ and $j$.

%
%
%
%

After each edge has been assigned  a channel, a channel-wise consensus iteration (Step 2 in Algorithm~\ref{alg:moment_method}) is performed amongst the agents. In particular, for a given channel $k\in S$, each agent $i$ updates the buffer $r_i(k)$ by using the information only from the neighbors in $\{j \in \mathcal N_i: c_{ij}=k\}$, i.e., only performs a consensus iteration on the edges whose communication channels have been selected to be $k$.

Moreover, the Handshake procedure selects a channel  for each edge randomly and repeatedly at each time $t\in \mathbb Z+$. This guarantees that,
over any long enough time period, the dynamical graph corresponding to each channel will be jointly connected. 
As a result, the buffers on each channel will reach consensus asymptotically, and then so does the reconstructed state $x_i(t)$. 
Proposition~\ref{prop_2} shows that  Algorithm~\ref{alg:moment_method} achieves   average consensus asymptotically for all the agents.

%

\renewcommand{\algorithmicrequire}{\textbf{Input:}}
\renewcommand{\algorithmicensure}{\textbf{Output:}}
\algsetup{indent=12pt}

\begin{algorithm}[h]
\caption{Privacy-preserving consensus with security degree $p$.}
\label{alg:moment_method}
\begin{algorithmic}[1]
\REQUIRE Security degree $p=(p_1, p_2, \dots, p_N)$; Initial $x_i(0)$, $\forall i \in [N]$; Graph $\mathcal G=(\mathcal E, \mathcal V)$; Secret key $S \in [1, \kappa]^{\bar p}$, which has been successfully distributed to all nodes; Step size $\alpha$.

\STATE \textbf{Initialization:} For each node $i$, randomly  choose\footnotemark \,
a coefficient vector $a^{(i)}\in \mathbb [-1,1]^{p_i-1}$, and Handshake Pairs $(b^{(i)}_j, c^{(i)}_j)$, $\forall j\in \mathcal N_i$, where $b^{(i)}_j\in [-1,1]$ and $c^{(i)}_j \in S$.
Initialize the buffers as $r_{i}(k)=f_i(k,0)$, for any $i\in [N]$, and $k\in S$.


\FOR{$t=1,2, \dots$}

\FOR{$i=1$  \TO $N$}
    \vspace{2mm} 
	\STATE \%\textit{Step 1 (Handshake): decide a channel $c_{ij}$ for each 	edge $(i,j)\in \mathcal E$.}
	\FOR{$j\in \mathcal N_i$}
	
		 \IF{$b^{(i)}_j\geq b^{(j)}_i$} 
	 	\STATE $c_{ij}\gets c^{(i)}_j$;
	 	\ELSE \STATE $c_{ij}\gets c^{(j)}_i$;
	 \ENDIF

%


	\ENDFOR
	\vspace{2mm}
	\STATE \%\textit{Step 2 (Chanel-wise consensus): update  $i$'s buffer using the information over channel $c_{ij}$ for each neighbor $j$. 
	Note that $r_j(c_{ij})$ is the only  information needed by $i$.}
	\FOR{$k$ in $S$}
	\STATE 
	\begin{equation}\label{eq:consensus_chanel}
	r_i(k) \!\gets\! r_i(k)\! + \!\alpha\sum_{j \in 		\mathcal N_i} \mathbbm{1}(c_{ij}\!=\!k) 					\big[r_j(k)\!-\!r_i(k)\big];  
	\end{equation} 
	\ENDFOR
	
	 \vspace{2mm}
	\STATE \%\textit{Step 3: Reconstruct the state $x_i(t)$ from buffers.}
	\STATE Update $x_i(t)$: 
	\[x_i(t)\gets \sum_{k\in S} r_i(k) \frac{\prod_{\ell\in S \setminus \{k\}}  (-\ell)}{ \prod_{\ell\in S \setminus \{k\}}  (k-\ell)  };
	\]

\ENDFOR

\vspace{2mm}
	\STATE \%\textit{Reset Handshake pairs.}
	\STATE For each node $i$, randomly  choose Handshake Pairs $(b^{(i)}_j, c^{(i)}_j)$, $\forall j\in \mathcal N_i$.
\vspace{2mm}

\ENDFOR

\end{algorithmic}
\end{algorithm}
\footnotetext{The random selection is drawn from the uniform distribution, in particular  $a^{(i)} \sim U([-1,1]^{p_i-1})$, $b^{(i)}_j \sim U([-1,1])$, and $c_j^{(i)}\sim U(S)$, where $U(\cdot)$ is the uniform distribution.}

\begin{proposition}\label{prop_2}
If the communication graph $\mathcal G$ is connected, and the step size $\alpha$ satisfies
\[\alpha < \frac{1}{\max_{i\in \mathcal V} |\mathcal N_i|},\] then Algorithm 1 solves Problem 1.
\end{proposition}
\begin{proof}
  For any time $t \in \mathbb Z+$, through the Handshake procedure, we denote by $c_{ij}(t) \in S$ the assigned communication channel on the edge $(i,j)$. 

  Then given a key  sequence $S\in \mathbb Z^{\bar p}$, for each channel $k \in S$, we define a channel vector 
  \[\xi_k(t)=\big(\,r_1(k), r_2(k),  \dots, r_N(k)\,\big) \bigg|_{t} \in \mathbb R^{N},\]
  where $r_i(k)\big|_t$ is agent $i$'s buffer corresponding to the channel $k$ at time $t$. Then this channel vector satisfies 
\begin{equation}\label{eq:proof_consensus_chanel_k}
\xi_k(t+1) = \xi_k(t) - \alpha L_k(t) \xi_k(t),
\end{equation}
where the matrix $L_k(t)$ is the Laplacian matrix of channel $k$  at time $t$. In particular, the matrix $L_k(t)=[l_{k,t}(i,j)]_{i,j} \in \mathbb R^{N \times N}$ satisfies
\[
l_{k,t}(i,j)\!=\!
\begin{cases}
-1 & \text{ if } i\!\neq\! j, \text{ and } c_{ij}(t)=k,\\
\sum_{\ell \in \mathcal{N}_i}\!\! \mathbbm{1}(c_{i\ell}(t)=k) & \text{ if } i\!=\! j,\\
0& \text{ otherwise}.\\ 
\end{cases}
\]

Denote the set of all the subgraphs of $\mathcal G$ by $\bar{\mathcal{G}}$. Then, for each channel $k$, define the dynamical graph  $\mathcal H_k : \mathbb Z+ \to \bar{\mathcal{G}}$, such that the Laplacian matrix of $\mathcal H_k (t)$ is  $L_k(t)$.
Then for all $k\in S$, dynamical graph $\mathcal H_k (t)$ satisfies that for any time $t_0$, there exists an infinite subset of time instants, denoted by $\bar T_{t_0} \subset \{t\in \mathbb Z+: t\geq t_0\}$, such that $\mathcal H_k(t)$ is connected  for any $t\in \bar T_{t_0} $. The   existence of $\bar{T}_{t_0}$ is because, in the Handshake procedure, the channels $c_{ij}(t)$ are randomly selected at each time $t$.

Then, for each channel $k$, we define the average buffer vector $\bar \xi_k(t)=\frac{1}{N} \mathbf 1^T \xi_k(t)$, which can be shown to be invariant along the trajectory  of system \eqref{eq:proof_consensus_chanel_k}. This is due to the fact that $\mathbf 1^T (I-\alpha L_k(t)) = \mathbf 1^T$, for any $t\in \mathbb Z+$. This implies 
\begin{equation}\label{eq:proof_constant_average}
\bar{\xi}_k(t)=\frac{1}{N} \mathbf 1^T \xi_k(0),\quad \forall t. 
\end{equation}
Then the consensus error is given by $\delta_k(t)=\xi_k(t)- \mathbf 1 \bar{\xi}_k(0)$. Note  that $\delta_k(t)^T \mathbf 1=0$ for any $t$. 
Furthermore, for any $t_p $ such that $L_k(t_p)$ corresponds to a connected graph, we know  that the matrix $ L_k(t_p)$ has a simple smallest eigenvalue $0$ with the eigenvector $\mathbf 1$.  This implies that
\begin{align}
&\max_{x\in \{x^T\mathbf 1=0\} }\frac{x^T (I-\alpha L_k(t_p)) x}{x^Tx} \\
=&1-\alpha\min_{x\in \{x^T\mathbf 1=0\} }\frac{x^T  L_k(t_p) x}{x^Tx}\nonumber \\
= &1-\alpha \lambda_2^k(t_p),\label{eq:proof_max_lamda2}
\end{align}
where $\lambda_2^k(t_p)$ is the second smallest eigenvalue of the Laplacian $L_k(t_p)$.

Then we consider a discrete Lyapunov function $\Phi_k(t)= \delta_k(t)^T\delta_k(t)$, for each $k\in S$. Let $t_0=1$. Then for any time instant $t_p \in \bar{T}_{t_0}$, we have 
\begin{align*}
\Phi_k(t_{p}+1) &=\|\big[I-\alpha L_k(t_{p})\big]\xi_k(t_p) - \mathbf 1 \bar{\xi}_k(0)\|^2_2\\
 &=\|\big[I-\alpha L_k(t_{p}) \big] \delta_k(t_p)\|^2_2\\
&\leq (1-\alpha \lambda_2(t_p) )^2 \|\delta_k(t_p)\|^2_2,
\end{align*} 
where we use $L_k(t_p)\mathbf 1=0$ in the second equality, and equation \eqref{eq:proof_max_lamda2} in the last inequality. By Gershgorin theorem, all eigenvalues of $L_k(t_p)$ are located  in the interval $[0, 2 \max_{i}|d_i^k(t_p) |]$, where $d_i^k(t_p)$ is the $i$-th diagonal element of Laplacian $L_k(t_p)$. Thus if $\alpha <\frac{1}{\max_{i\in \mathcal V} |\mathcal N_i|} < \frac{1}{\max_{i\in \mathcal V} |d_i^k(t_p)|}$, then $| 1 - \alpha \lambda_2(t_p)| <1$. Therefore, the Lyapunov function $\Phi_k(t)$ is strictly decreasing at any time $t\in \bar T_{t_0}$. Combine the argument \eqref{eq:proof_constant_average},  average consensus follows for all channel $k\in S$.

Then we have for any  $k\in S$ and $i\in \mathcal V$ ,
\begin{align*}
\lim_{t\to \infty} r_i(k)\big|_t &= \frac{1}{N}\sum_{i\in \mathcal V} r_i(k)\big|_{t=0}= \frac{1}{N}\sum_{i\in \mathcal V} f_i(k,0)\\
&= \frac{1}{N}\sum_{i\in \mathcal V}\sum_{j=1}^{p_i-1} \big(  a_j^{(i)}\big) k^j + \frac{1}{N}\sum_{i\in \mathcal V} x_i(0)\\
&=\sum_{j=1}^{\bar p-1} \Bigg(\frac{1}{N}\sum_{i\in \mathcal V}   a_j^{(i)}\Bigg) k^j + \frac{1}{N}\sum_{i\in \mathcal V} x_i(0),
\end{align*}
in which we use $\bar p\geq \max_{i\in [N]} p_i$, and  we set $a^{(i)}_j=0$ for any $j \geq p_i$.
Then we see that after   consensus is reached, i.e., when $t$ is large enough,  every node obtains $\bar p$ observations $\{(k,F(k)): k\in S\}$ of the average polynomial $F(\theta)$, where $F(\theta)$ is 
\[
F(\theta) =\sum_{j=1}^{\bar p -1} \Bigg(\frac{1}{N}\sum_{i\in \mathcal V}   a_j^{(i)}\Bigg) \theta^j + \frac{1}{N}\sum_{i\in \mathcal V} x_i(0).
\]
According to the Shamir's algorithm, the constant term $\frac{1}{N}\sum_{i\in \mathcal V} x_i(0)$ can be reconstructed by each agent $i$ as
\[
\frac{1}{N}\sum_{i\in \mathcal V} x_i(0)=\sum_{k\in S} r_i(k)\big|_{t=\infty} \frac{\prod_{\ell\in S \setminus \{k\}}  (-\ell)}{ \prod_{\ell\in S \setminus \{k\}}  (k-\ell)  }=x_i(\infty).
\]
Thus the assertion holds. 

The $p$-degree security of Algorithm 1 follows that for any agent $i\in \mathcal V$.
Moreover, 
any collusion of less than its $p_i$ neighbors cannot reconstruct $x_i(t)$. 

For any adversary, since the key sequence $S$ is unknown, the interception on the communication link does not give any useful information to reconstruct $x_i(t)$.
$\blacksquare$
\end{proof}

Note that in Algorithm~\ref{alg:moment_method}, 
we can use an alternative update law for each $r_i(k)$ in \eqref{eq:consensus_chanel}, which is 
\[r_i(k) \gets \alpha_i r_i(k) + \alpha_i\sum_{j \in \mathcal N_i} \mathbbm{1}(c_{ij}=k) r_j(k),\]
where $\alpha_i = \frac{1}{1+\sum_{j\in \mathcal N_i} \mathbbm{1}(c_{ij}=k)}$. But by Lemma~\ref{lem:time-switching_nonaverage_consensus}, this update law can only achieve consensus, but the reached state does not necessarily equal to $\frac{1}{N}\sum_{i\in \mathcal V} x_i(0)$.

\subsection{Key distribution}\label{sec:key_distribution}
To protect the privacy-preserving algorithm from eavesdropping, a the  secret key sequence $S$ must first be distributed to each individual in the network. By using such a key sequence, the information exchanged between the agents afterwards becomes resistant to the wiretapping on communication links.

Indeed, the key sequence can be set to some default value, e.g., $S=(1, 2, \dots, \bar p)$. Even with such  kinds of default keys being public, any eavesdropper would still need at least $p_i$ pieces of information transmitted to/from agent $i$ to reconstruct the state $x_i(t)$. However, we still provide here a key distribution algorithm for the scenario in which the key $S$ must be kept secret.

We note that  consensus algorithms typically only provide  asymptotic convergence to a synchronized state. This means that, within finite time, one agent's key updated by the algorithm can only approach, but not equal, the keys of the others.
To reach an exact synchronization for the key sequence,  finite-time consensus algorithms are needed, such as, e.g., \cite{kibangou2014step,sundaram2007finite,wang2010finite,wei2018finite}. 
These finite-time consensus algorithms, however, often depend on particular communication graphs \cite{kingston2006discrete}, or require first computing a matrix factorization  \cite{kibangou2014step} or minimal polynomial \cite{sundaram2007finite} of the weight matrix. For our purpose, with the elements of the key sequence being integers, we can adopt the following finite-step consensus algorithm that can be computed much more easily.

Define the communication weight matrix $W=[w_{ij}] \in \mathbb R^{N\times N}$, satisfying the following conditions:
\begin{enumerate}
\item[(a)] $w_{ij}>0$ if $(i,j)\in \mathcal E$, and $w_{ij}=0$ if $(i,j) \notin \mathcal E$.  Moreover, $w_{ii}>0$ for $i\in \mathcal N$.
\item[(b)] The matrix $W$ is doubly stochastic, i.e., $\sum_{k} w_{ki}=1$, and $\sum_{k} w_{ik}=1$ for any $i\in \mathcal N$.
\end{enumerate}

Note that for any connected graph $\mathcal G$, the nonnegative matrix $W$ is irreducible, and moreover, $W$ is primitive, i.e., $W$ has only one eigenvalue with maximum modulus. In particular, we have the largest eigenvalue of $W$, which is $1$, is simple  with eigenvector $\mathbf 1$.

Denote the spectral norm \footnote{Spectral norm of a matrix $A\in \mathbb C^{n\times n}$ is defined by \[\max_{\|x\|_2 \neq 0} \frac{\| A x\|_2}{\|x\|_2},\] which is equal to the square root of the maximum eigenvalue of  $A^*A$.}
 of $W-\frac{1}{n}\mathbf 1 \mathbf 1^T$ by $\gamma$. Then we have $\gamma \in (0, 1)$. This is because the positive semidefinite matrix $W^TW$ has a unique largest eigenvalue $1$ with eigenvector $\mathbf 1$. Consequently,, all the eigenvalues of $(W- \frac{1}{n}\mathbf 1 \mathbf 1^T)^T (W- \frac{1}{n}\mathbf 1 \mathbf 1^T) = W^TW- \frac{1}{n} \mathbf 1 \mathbf 1^T$ belong to $(0,1)$.
 
We employ the classical consensus algorithm 
\begin{equation}\label{eq_clasical_consensus}
S_i(t+1) = \sum_{j=1}^N w_{ij} S_{j}(t),
\end{equation}
for $S_i(t) \in \mathbb R^{\bar p}$  and $i\in[N]$. Then the following lemma gives the convergence rate for the consensus dynamics \eqref{eq_clasical_consensus}.

\begin{lemma}\label{lem:convergence_rate}
Let $\kappa \in \mathbb R+$, and let a threshold $\delta \in \mathbb R+$ be given. Then, for any initial condition  $S_i(0) \in [0, \kappa]^{\bar p}$, $i\in [N]$, the trajectory of the consensus algorithm \eqref{eq_clasical_consensus} satisfies 
\[
\|S_i(T) - \frac{1}{N} \sum_{j=1}^N S_j(0)\|_\infty \leq \delta, \qquad \forall i\in [N],
\] 
where $T = \ceil*{\log_{\lambda} \frac{\delta}{\kappa \sqrt{ N} }}$.
\end{lemma}

\begin{proof}
See Appendix B.
\end{proof}

Next, we  give the algorithm to distribute an integer key sequence to all the agents. Note that in the algorithm we denote by $S_{i\ell}$ the $\ell$-th element of   agent $i$'s key sequence $S_i$.
\begin{algorithm}[h]
\caption{Finite-step key sequence distribution.}
\label{alg:key_distribution}
\begin{algorithmic}[1]
\REQUIRE $\bar p$: the number of channels; $\bar N$: the upper bound for the number of nodes; $\lambda$: the spectral norm  of $W-\frac{1}{n}\mathbf 1 \mathbf 1^T$.

\STATE \textbf{Initialization:} 
For each node $i$, randomly
 choose a real vector $S_i(0) \sim U( [1 , \kappa]^{\bar p})$, where $\kappa$ is a default constant such that $\bar p \ll \kappa$. $T= 0.$

\STATE flag $\gets \TRUE$  \quad \%  \textit{Indicator that the key distribution is incomplete.}

\WHILE{flag  $= \TRUE$} 

\STATE \%\textit{Consensus for  $\ceil*{\log_{\lambda} \frac{0.1}{\kappa \sqrt{\bar N} } }$ steps.}
\FOR{$t=T+1, \dots, T+\ceil*{\log_{\lambda} \frac{0.1}{\kappa \sqrt{\bar N} } } $}

\FOR{$i=1$  \TO $N$}
\STATE $S_i(t+1) = \sum_{j=1}^N w_{ij} S_{j}(t).$
\ENDFOR
\ENDFOR
\STATE flag $\gets \FALSE$ 
\vspace{2mm}
\STATE \%  \textit{Check if the key distribution is done.}
\FOR{$i=1$  \TO $N$}
  \FOR{$\ell = 1$ \TO $\bar p$}
  \STATE $\Delta_1 \gets S_{i\ell}(t) - \floor*{S_{i\ell}(t)}$ \%  \textit{decimal part.}
  \STATE $\Delta_2  \gets  \min_{j\in[1,\ell)} \big|\floor*{S_{i\ell}(t)} - \floor*{S_{ij}(t)}\big| $ \%  \textit{minimal discrepancy of element $S_{i\ell}$ away from the elements before it.}
  \IF{ $\Delta_1 \leq 0.1$ \OR $\Delta_1 \geq 0.9$ \OR $\Delta_2 <0.5 $ }
	\STATE $S_{i\ell}(t) \gets y$,  for a random integer $y \in [1, \kappa]$.  
	\STATE flag $\gets \TRUE$    
  \ENDIF
  
  \STATE $S_{i\ell}(t) = \floor*{S_{i\ell}(t)}.$
  
  \ENDFOR
\ENDFOR

\ENDWHILE

\end{algorithmic}
\end{algorithm}

By Lemma~\ref{lem:convergence_rate},  Algorithm~\ref{alg:key_distribution} first guarantees that after a number  $\ceil*{\log_{\lambda} \frac{0.1}{\kappa \sqrt{\bar N} } }$ of consensus iterations, the consensus error  from the average is smaller than the threshold $\delta=0.1$. Then if each element in agent $i$'s state keeps some distance away from the closest integer, and if there are no two identical elements in agent $i$'s key sequence $S_i$, we take
the integer part of the key sequence $S_i$ as $i$'s key sequence.

Although, using  Algorithm~\ref{alg:key_distribution}, we do not need  to compute a matrix factorization, or the minimal polynomial  of the weight matrix \cite{kibangou2014step,sundaram2007finite}, the algorithm requires some global information, i.e., the spectral norm $\gamma$ of the graph. To avoid this, we can employ the  Metropolis matrix as the weight matrix, which is defined by
\[
w_{ij}=
\begin{cases}
\frac{1}{2\max(d_i, d_j)}, \quad &\text{if } (i, j) \in \mathcal E \text{ for } i\neq j, \\
1- \sum_{\ell \in \mathcal N_i} \frac{1}{2\max(d_i,d_\ell)}, \quad &\text{if } i=j,\\
0, \quad &\text{otherwise},
\end{cases}
\]
where $d_i$ is the degree of node $i$, i.e., $d_i=|\mathcal N_i|$. By using this weight matrix, each node is only required to know the set of its neighbors. Moreover, in this case, the spectral  norm of the matrix $W-\frac{1}{n}\mathbf 1 \mathbf 1^T$ satisfies $\gamma \leq 1-\frac{1}{71 N^2}$ \cite[Lemma 2.2]{olshevsky2014linear}.

Actually, when the number of nodes is large, the convergence rate of the consensus algorithm, or the computation complexity for finite-step consensus algorithms becomes enormous. Therefore, one solution to this is using a default (or even public)  key. By the results in section \ref{sec:PPC},  we  see that even if an adversary is aware of the secret key, the privacy of agent $i$  is not disclosed if the number of $i$'s communication channels eavesdropped by the adversary is less than $p_i$.

\section{Simulation}
In this section, we will illustrate the  algorithm using a numerical example, in which a network consisting of $5$ nodes with a connection topology of cyclic graph achieves average consensus using the proposed privacy-preserving method. 
As reported in \cite{olfati2007consensus}, if the step size in the synchronization algorithm is not carefully chosen the system under a cyclic  connection graph will only reach  consensus but not  average consensus.

Specifically, we consider a cyclic graph $\mathcal G_c(5)$ consisting of $5$ nodes as shown in Figure~\ref{Fig:M1}. Let the upper bound of the security degree $\bar p = 4$, and set the security degree $p=(2, 3, 4, 2, 3 )$.
\begin{figure}[h]
    \centering
    \includegraphics[width=0.3\textwidth]{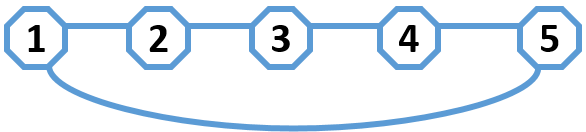}
    \caption{A cyclic graph $\mathcal G_c(5)$ with $5$ nodes.}\label{Fig:M1}
\end{figure}

Firstly a secret key $S\in [1, 20]^4$ is distributed using Algorithm~2, after $28$ iterations all agent share a common secret key
\[
S=(4,7,15,3).
\] 
Then  Algorithm~1 is used to achieve  average consensus with the common secret key $S$. In Figure~\ref{Fig:M1}, we show the trajectories of the buffer $r_i(4)$ for all five agents. Figure~\ref{Fig:M2} presents the states $x_i(t)$ reconstructed from each channel's buffer. We can see that  consensus is reached, and moreover, that consensus state is equal to the average of the initial states. 

\begin{figure}[h]
    \centering
    \includegraphics[width=0.5\textwidth]{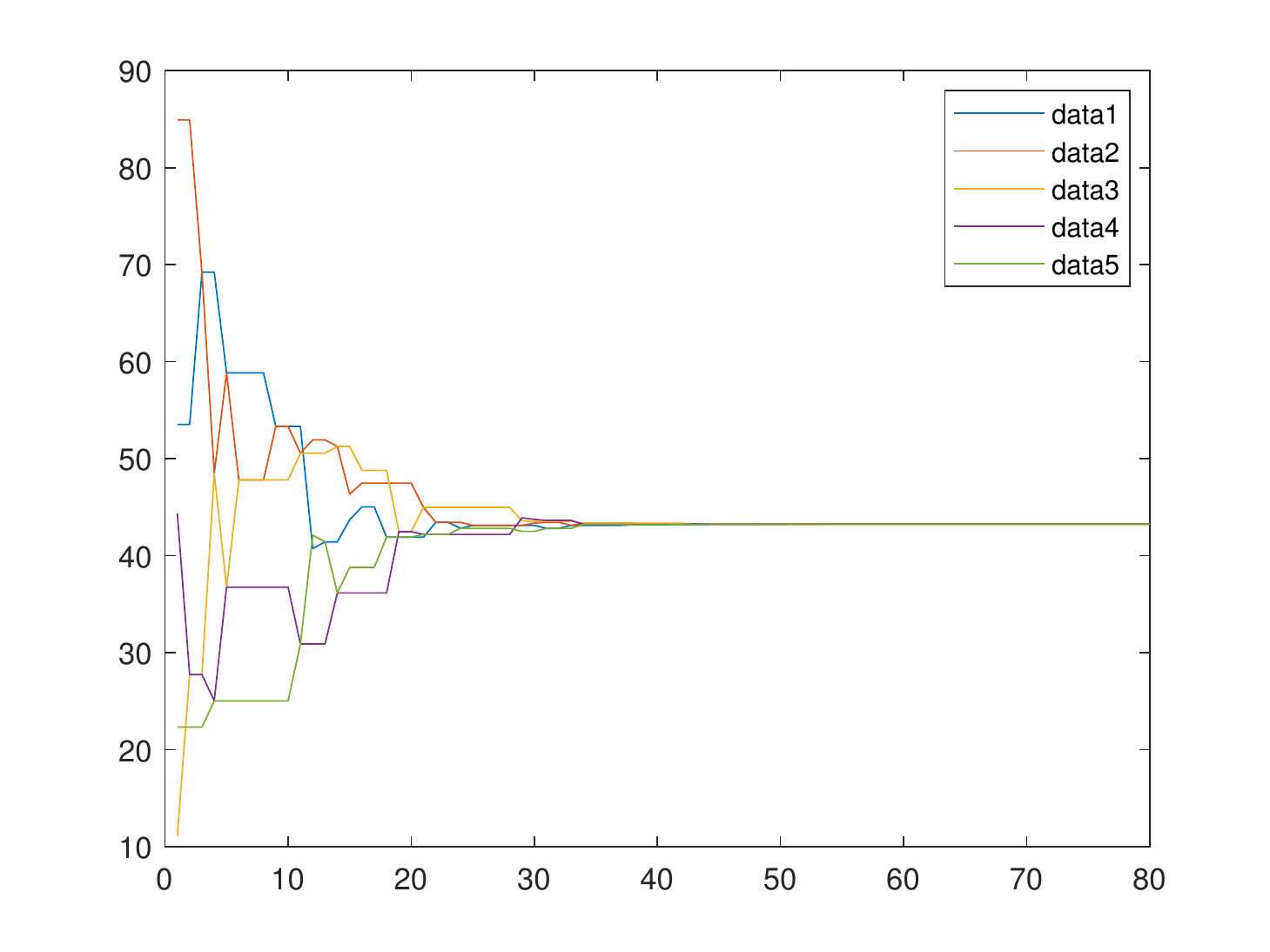}
    \caption{The   buffer trajectories of the five agents, corresponding to the first channel $k = 4$.}\label{Fig:M1}
\end{figure}

\begin{figure}[h]
    \centering
    \includegraphics[width=0.5\textwidth]{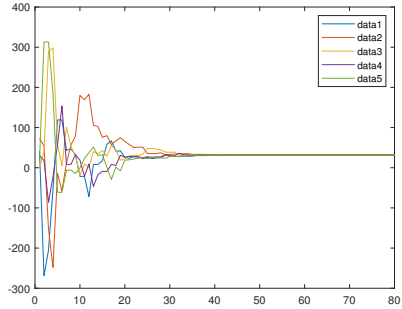}
    \caption{The trajectories of the reconstructed states of all nodes using Algorithm 1. The average of the initial states is $37.731$, and the consensus state is $31.731$.}\label{Fig:M2}
\end{figure}


\section{Conclusion}
This paper proposes a privacy-preserving mechanism for the average-consensus problem based on the secret sharing scheme. The proposed algorithm renders the network resistant to the collusion of any given number of neighbors, and protects the consensus procedure from eavesdropping. In  future work, we will extend this idea to formation control and distributed optimization, and also exploit it in the relevant applications, such as opinion agreement, sensor network averaging, survey mechanism, and distributed decision making.

\section*{Appendix A: Shamir's scheme}
The Shamir's secret sharing scheme \cite{shamir1979share} allows a user to ``share" a secret $M$ among a group of $n$ participants, such that (a) any $p$ or more of the participants can reconstruct $M$, and (b)any set of less than $p$ participants learn nothing about $M$.

For a secret $M \in \mathcal F$, where $\mathcal F$ is a field, the algorithm Share($S$) is 

\begin{enumerate}
\item[(1)] Let $a_0=M$;
\item[(2)] Choose at random $a_1,\dots, a_{p-1}\in \mathcal F$;
\item[(3)] Let $f(x)=\sum_{i=0}^{p-1} a_i x^i$;
\item[(4)] The secret share $M_i=(x_i,y_i)$ for any $i\in [n]$, where $y_i=f(x_i)$.
\end{enumerate}

To reconstruct $M$, we need at least $p$ shares of secret  $\{M_j=(x_j,y_j)\}_{j=1}^p$. Then one can use the Lagrange interpolation formula, 
\[
f(x)=\sum_{j=1}^p y_j L_j(x),
\]
where the Lagrange basis polynomial $L_j(x)$ is defined by
\[
L_j(x)=\frac{\prod_{k\in [p] \setminus \{j\}}  (x-x_k)}{ \prod_{k\in [p] \setminus \{j\}}  (x_j-x_k)  },
\]
for each $j\in [p]$.
Once the polynomial $f(x)$ has been reconstructed, the secret message can be recovered as $M= f(0)$, due to the definition of $f(x)$. Equivalently, the true message $M$ can be reconstructed by
\[
M=f(0)=\sum_{i=1}^M y_i L_i(0)=\sum_{j=1}^M y_j \frac{\prod_{k\in [p] \setminus \{j\}}  (-x_k)}{ \prod_{k\in [p] \setminus \{j\}}  (x_j-x_k)  }.
\]

\section*{Appendix B: Proof of Lemma~\ref{lem:convergence_rate}}
For each $\ell \in [\bar p]$, we denote the $\ell$-component vector  by
\[
S^{(\ell)}(t)= \big(S_{1\ell}(t), S_{2\ell}(t), \dots, S_{N\ell}(t) \big) \in \mathbb R^N, 
\]
where $S_{ij}(t)$ is the $j$-th element of the key $S_i(t)$. Then  for each $\ell$, the consensus iteration reads $S^{(\ell)}(t+1)= WS^{(\ell)}(t)$.

Next, the average for element $\ell$ is $\alpha_\ell(t)= \frac{1}{N}\mathbf 1^T S^{(\ell)}(t)$. Moreover, for any $t$, $\alpha_\ell(t+1)=\alpha_\ell(t)$, due to the fact that $\mathbf 1^T W=\mathbf 1^T$. Thus we can denote the invariant average for component $\ell$ by $\alpha_\ell=\alpha_\ell(0)$. Then we have for each $\ell \in [\bar p]$,
\begin{align*}
\|S^{(\ell)}(t+1) -  \alpha_\ell \mathbf 1 \|_\infty 
&\leq
\|S^{(\ell)}(t+1) -  \alpha_\ell \mathbf 1 \|_2\\
&=\|WS^{(\ell)}(t) -  \alpha_\ell \mathbf 1 \|_2 \\
&= \|(W-\frac{1}{n}\mathbf 1 \mathbf 1^T)(S^{(\ell)}(t) -  \alpha_\ell \mathbf 1) \|_2 \\
& \leq \lambda \|S^{(\ell)}(t) -  \alpha_\ell \mathbf 1 \|_2
\end{align*}
Then for any time $t$, $\|S^{(\ell)}(t) -  \alpha_\ell \mathbf 1 \|_\infty \leq \lambda^t \|S^{(\ell)}(0) -  \alpha_\ell \mathbf 1 \|_2 \leq \lambda^t \sqrt {N}\|S^{(\ell)}(0) -  \alpha_\ell \mathbf 1 \|_\infty$.  According to the initial condition, $\|S^{(\ell)}(0) -  \alpha_\ell \mathbf 1 \|_\infty \leq \kappa$. Then the assertion follows. $\blacksquare$

\bibliographystyle{ieeetr}
\bibliography{ref}

\vspace{10mm}
\textbf{Silun Zhang}  received his B.Eng. and M.Sc.
degrees in Automation from Harbin Institute of
Technology, China, in 2011 and 2013 respectively,
and the PhD degree in Optimization and Systems
Theory from Department of Mathematics, KTH
Royal Institute of Technology, Sweden, in 2019.

He is currently a Wallenberg postdoctoral fellow with the
Laboratory for Information and Decision Systems
(LIDS), MIT, USA. His main research interests
include nonlinear control, networked systems, rigidbody attitude control, and modeling large-scale systems.

\vspace{5mm}
\textbf{Thomas Ohlson Timoudas} received his PhD degree in Mathematics in October 2018 from KTH Royal Institute of Technology, Sweden, and his MSc (2013) and BSc (2012) degrees in Mathematics from Stockholm university, Sweden.

He is currently a postdoctoral researcher with the Department of Network and Systems Engineering at KTH Royal Institute of Technology, Sweden. His main research interests include dynamical systems, networked systems, internet of things, and distributed algorithms.

\vspace{5mm}
\textbf{Munther A. Dahleh} received his Ph.D. degree from Rice University, Houston, TX, in 1987 in Electrical and Computer Engineering. Since then, he has been with the Department of Electrical Engineering and Computer Science (EECS), MIT, Cambridge, MA, where he is now the William A. Coolidge Professor of EECS. He is also a faculty affiliate of the Sloan School of Management. He is the founding director of the newly formed MIT Institute for Data, Systems, and Society (IDSS). Previously, he held the positions of Associate Department Head of EECS, Acting Director of the Engineering Systems Division, and Acting Director of the Laboratory for Information and Decision Systems. He was a visiting Professor at the Department of Electrical Engineering, California Institute of Technology, Pasadena, CA, for the Spring of 1993. He has consulted for various national research laboratories and companies.

Dr. Dahleh is interested in Networked Systems with applications to Social and Economic Networks, financial networks, Transportation Networks, Neural Networks, and the Power Grid. Specifically, he focuses on the development of foundational theory necessary to understand, monitor, and control systemic risk in interconnected systems.  He is four-time recipient of the George Axelby outstanding paper award for best paper in IEEE Transactions on Automatic Control. He is also the recipient of the Donald P. Eckman award from the American Control Council in 1993 for the best control engineer under 35. He is a fellow of IEEE and IFAC.

\end{document}